\documentclass[letterpaper]{IEEEtran}
\usepackage[utf8]{inputenc} 
\usepackage[T1]{fontenc}
\usepackage{url}              
\usepackage{ifthen}          
\usepackage{cite}             
\usepackage[cmex10]{amsmath}  
\usepackage{color}
\usepackage{mleftright}       
\mleftright                   

\usepackage{graphicx}         
\usepackage{booktabs}         
\usepackage{hyperref}
\usepackage{comment}
\usepackage{amsthm}
\usepackage{amsfonts}
\usepackage{bbm}
\usepackage[square,numbers]{natbib}

\newtheorem{theorem}{Theorem}
\newtheorem{lemma}[theorem]{Lemma}

\newtheorem{corollary}[theorem]{Corollary}
\newtheorem{conjecture}[theorem]{Conjecture}
\theoremstyle{definition}
\newtheorem{definition}{Definition}

\newtheorem{remark}{Remark}

\newcommand{\E} {\mathbb{E}}

\renewcommand{\P} {\mathbb{P}}

\DeclareMathOperator*{\argmax}{arg\,max}

\newcommand{\beq}{ \begin{equation} }
	\newcommand{\eeq}{ \end{equation} }

\begin{document}
\title{Graph Matching in Correlated Stochastic Block Models for Improved Graph Clustering} 


%
 \author{%
   \IEEEauthorblockN{Joonhyuk Yang and Hye Won Chung}\\
   \IEEEauthorblockA{School of Electrical Engineering\\
                     KAIST\\
                      \{joonhyuk.yang, hwchung\}@kaist.ac.kr}
                   }

\maketitle

\thispagestyle{empty}


\begin{abstract}
We consider community detection from multiple correlated graphs sharing the same community structure.
The correlated graphs are generated by independent subsampling of a parent graph sampled from the stochastic block model. The vertex correspondence between the correlated graphs is assumed to be unknown. We consider the two-step procedure where the vertex correspondence between the correlated graphs is first revealed, and the communities are recovered from the union of the correlated graphs, which becomes denser than each single graph. We derive the information-theoretic limits for exact graph matching in general density regimes and the number of communities, and then analyze the regime of graph parameters, where one can benefit from the matching of the correlated graphs in recovering the latent community structure of the graphs. \footnote{This research was supported by the National Research Foundation of Korea under grant 2021R1C1C11008539, and  by the MSIT, Korea, under the ITRC support program (IITP-2023-2018-0-01402) supervised by the IITP.}
\end{abstract}

\section{Introduction}\label{sec:Introduction}

Community detection, or graph clustering, is a core problem in data mining or machine learning, where the dataset is given in the form of a graph composed of vertices connected to each other based on the levels of interactions or similarity, and the goal is to partition the vertices into densely connected clusters. 
The stochastic block model (SBM) \citep{HLL83} is widely used as a generative model for community detection. 
In the SBM, $n$ vertices are partitioned into $k$ communities (either randomly or deterministically), and the vertex pairs from the same community are connected by an edge with probability $p\in[0,1]$, while the vertex pairs from different communities are connected with a lower probability $q\in[0,p)$.

One of the fundamental questions in the community detection is to identify the regime of parameters where the community detection is information-theoretically feasible/infeasible. 
When the community label of each vertex is uniformly sampled from $[k]$, given a graph $G\sim \text{SBM}(n,p,q,k)$ with $p=\frac{\alpha \log n}{n}$, $q=\frac{\beta \log n}{n}$ for positive constants $\alpha,\beta$ and $k=\Theta(1)$,  it was shown in  \cite{AS15,ABH15,Abbe17} that if $ \sqrt{\alpha}-\sqrt{\beta}>\sqrt{k}$ then the exact community recovery is possible. Moreover, it was shown that this limit is achievable by an efficient algorithm, such as Semidefinite Programming (SDP) \citep{CX16,HWX16a,HWX16b,ABKK17,LCX21}. Especially, the authors in \cite{ABKK17} found the precise threshold for community recovery with the SDP even for the case where $k$ increases in $n$. Suppose that the size of each community is the same as $m=\frac{n}{k}$ and the intra-/inter-community edge density is $p=\frac{\alpha \log m}{m}$ and $q=\frac{\beta \log m}{m}$, respectively, for positive constants $\alpha$ and $\beta$. Then, if $\sqrt{\alpha}-\sqrt{\beta}>1$, the community recovery is possible by an SDP as long as $k=o(\log n)$. On the other hand, the exact community recovery is impossible with any algorithm if $\sqrt{\alpha}-\sqrt{\beta}<1 $ for $k=n^{o(1)}$. 

The main question we address in this paper is whether one can improve this information-theoretic limit on the community recovery when an additional graph correlated with the original graph is given as side information. For this scenario, assume that there exists a parent graph $G$, and two subgraphs $G_1$ and $G_2'$ are obtained by independent subsampling of $G$, where every edge of $G$ is sampled independently with probability $s>0$. For a fixed permutation $\pi_*:[n]\to[n]$, let $G_2$ be the graph obtained by permuting the vertices of $G_2'$ by $\pi_*$. Then, can one benefit from having a correlated graph $G_2$ in addition to $G_1$ in recovering the community structure of $G_1$?

This problem has a close connection to the graph matching, or graph alignment, which has been widely studied with applications of social network analysis, bioinformatics, and machine learning. When the vertex correspondence between the two correlated graphs $G_1$ and $G_2$ can be revealed by graph matching, by recovering $G'_2$ of which the vertices are indexed in the same order as $G_1$, we can obtain a denser graph $G_1\cup G'_2$ and recover the community structure of $G_1$ from $G_1\cup G'_2$, even when it is impossible to do so by using $G_1$ alone. This problem has been studied in \cite{RS21} for two balanced communities where the average degree is logarithmic in the number of vertices.

We generalize this result by identifying the information-theoretic limits for graph matching in density $nps^2,nqs^2\gg 1$, $p,q=o(1)$ and for $k$ scaling in $n$ (Thm. \ref{thm:matching achievability} and \ref{thm:matching impossible}). Additionally, for the equal-sized communities of size $m=n/k$ and $p=\frac{\alpha \log m}{m}$, $q=\frac{\beta \log m}{m}$, we show that the information-theoretic limit for community recovery can be improved from $\sqrt{\alpha}-\sqrt{\beta}>\sqrt{1/s}$ to $\sqrt{\alpha}-\sqrt{\beta}>\sqrt{1/(1-(1-s)^2)}$ by having a correlated graph $G_2$ in addition to $G_1$, when $k=o(\log n)$ and an extra condition of $s^2(\alpha+(k-1)\beta)>1$ for graph matching is satisfied (Cor. \ref{cor:community recovery} and \ref{cor:community recovery impossible}).

\subsection{Correlated Stochastic Block Models}\label{sec:model}
To consider the community recovery problem from multiple correlated graphs, we focus on the correlated stochastic block models \citep{OGE16}.
Assume a parent graph $G \sim$ SBM$(n,p,q,k)$. 
The community label of each vertex is uniformly sampled from $[k]$. Let $\sigma=(\sigma_1,\dots, \sigma_n)$ denote the community labels of $n$ vertices. 
 Two subgraphs $G_1$ and $G_2'$ of $G$ are generated by sampling every edge of $G$ independently with probability $s\in [0,1]$. Note that for any distinct $i,j\in[n]$, 
$
		\P\left((i,j)\in \mathcal{E}(G_1) |  
			(i,j)\in \mathcal{E}(G'_2) \right)
			= s,
$ where $\mathcal{E}(G)$ is the edge set of graph $G$.
For a fixed permutation $\pi_{*} : [n] \rightarrow [n]$, we obtain the graph $G_{2}$ by permuting all vertices of $G'_{2}$ with $\pi_{*}$. Our goal is to find the regime, where $\pi_{*}$ can be discovered, so that $G_1\cup G_2'$ can be obtained. 


\subsection{Related work and our contribution}\label{sec:prior work}

Graph matching has been widely studied in the correlated Erd\H{o}s-R\'enyi (ER) model \citep{PG11}, where the parent graph $G$ is sampled from ER model, $G \sim \mathcal{G}\left(n, p\right)$. 
In  \cite{CK16,WXY22}, the information-theoretic limit for exact matching was analyzed and it was shown that the exact matching is possible if $nps^2 > (1+\epsilon) \log n$ for any positive constant $\epsilon$. Furthermore, in \cite{GM20,WXY22,GML21,DD22}, the precise limit for partial recovery was analyzed and it was shown that partial recovery is possible if $nps^2 \geq \lambda_* +\epsilon$ for any positive constant $\epsilon$, where $p=n^{-a+o(1)}$ for $a \in(0,1)$ and  $\lambda_*$ is some specified constant \citep{DD22}.

In the correlated SBMs, the information-theoretic limit for exact matching has been studied mainly under the assumption that community labels are given. 
Cullina et al. \cite{CSKM16} proved that the exact matching is possible if $ \frac{\alpha+(k-1)\beta}{k}s^{2} >2 $ with $p=\frac{\alpha \log n}{n}$ and $q=\frac{\beta \log n}{n}$ for positive constants $\alpha,\beta$ and an integer $k$. Onaran et al. \cite{OGE16} showed that exact matching is possible if $
    s\left(1-\sqrt{1-s^{2}}\right)\frac{\alpha+\beta}{2}>3
$
for the case of two communities. However, in the case of $\alpha=\beta$, the correlated SBMs becomes the correlated ER model, where it has been known that the exact matching is possible if  $\alpha s^{2}>1$. Thus, we can see that the above bounds are not tight. Cullina et al. \cite{CSKM16} also showed that the exact matching is impossible if $ \frac{\alpha+(k-1)\beta}{k}s^{2} <1 $ where $k=O(n^c)$ for a sufficiently small $c$.


Recently, Racz and Sridhar \cite{RS21} found the precise information theoretic limit for exact matching for two balanced communities, matching the converse result, even when the community labels are not given. 
From the results of the correlated ER model \citep{WXY22} and the correlated SBMs \citep{RS21} with two communities, we can conjecture that if there is no isolated vertex in the intersection graph $G_{1} \cap G'_{2}$, then the exact matching may be possible. In this work, we indeed show that for general correlated SBMs with $k$ balanced communities, the exact matching is possible if the average degree is greater than $\log n$. We prove this with an additional condition on $k$ depending on $p,q$, and $s$. Conversely, we also show that when the average degree is less than $\log n$, it is impossible to exactly match the graphs due to the existence of isolated vertices on graphs $G_1 \cap G'_2$. In the case of converse, $k$ can be $O(n^{t})$  for a sufficiently small positive constant $t$. Moreover, we extend the result of \citep{RS21} and show that there exist regimes where the community recovery becomes easier when a correlated graph is given as side information when $k=o(\log n)$. 

\section{Main Result}\label{sec:Main result}
Asymptotic dependencies are denoted with standard notations $O(\cdot), o(\cdot), \Omega(\cdot), \omega(\cdot), \Theta(\cdot)$ with $n\to\infty$.  Let $\wedge$ and $\vee$ denote the $\min$ and the $\max$ operator.
\subsection{Exact graph matching}

\begin{theorem}[Achievability for exact matching]\label{thm:matching achievability}
	For $p,q = o(1)$ and $s \in[0,1]$ such that $
     nps^{2}, nqs^{2}=\omega(1),$ let $G_{1}$ and $G_{2}$ be two graphs given by the correlated SBMs defined in Sec \ref{sec:model}. There exist constants $C,c>0$ with the following property. For any arbitrarily small constant $\varepsilon>0$, if
	\begin{equation}\label{eq:thm matching degree}
		ns^{2}\left(\frac{p+(k-1)q}{k}\right)>(1+\varepsilon)\log n, \text{ and}
	\end{equation}
 \begin{equation}\label{eq:thm matching k}
     k \leq C\left(\sqrt{nps^{2}} \wedge nqs^{2}\wedge n^{c} \right),
 \end{equation}
	then there exists an estimator $\hat{\pi}$ such that $\hat{\pi} = \pi_{*}$ with probability at least $1-o(1)$.
\end{theorem}

\begin{theorem}[Converse for exact matching]\label{thm:matching impossible}
    For $p,q,s \in[0,1]$ such that $ps^{2}=o(1)$, let $G_{1}$ and $G_{2}$ be two graphs given by the correlated SBMs defined in Sec \ref{sec:model}. For any arbitrarily small constant $\epsilon>0$, there exists a positive constant $t$ depending on $\epsilon$ with the following property. If
	\begin{equation}\label{eq:thm:matching impossible degree}
		ns^{2}\left(\frac{p+(k-1)q}{k}\right)<(1-\epsilon)\log n, \text{ and}
	\end{equation}
    \begin{equation}\label{eq:thm:matching impossible k}
          k=O(n^{t}),
    \end{equation}
	then for any estimator $\hat{\pi}$, we have  $\lim _{n \rightarrow \infty} \P\left(\hat{\pi}=\pi_{*}\right)=0$.    
\end{theorem}

\begin{remark}[Gap between the achievability and converse] From Theorem \ref{thm:matching achievability} and \ref{thm:matching impossible}, it can be seen that there exists no gap for the density conditions between the achievability \eqref{eq:thm matching degree} and converse \eqref{eq:thm:matching impossible degree}, but there exists a gap between the conditions on $k$ \eqref{eq:thm matching k} and \eqref{eq:thm:matching impossible k}. It is open whether one can increase $k$ to $n^c$ for some constant $c>0$ in the achievability.
\end{remark}

\subsection{Exact community recovery}
Suppose that the community size is the same as $m=n/k$.
We will show that when two correlated SBMs  are given, community recovery becomes easier compared to the case of having only one graph through the matching of the two. By combining Theorem \ref{thm:matching achievability} with the information-theoretic limit from \cite{ABKK17}, we obtain the following result. 
\begin{corollary}\label{cor:community recovery}
    For $p=\frac{\alpha \log m}{m},\; q=\frac{\beta \log m}{m}$ with positive constants $\alpha, \beta$ and $s\in [0,1]$, let $G_{1}$ and $G_{2}$ be two graphs given by the correlated SBMs defined in Sec \ref{sec:model}. Suppose that 
    \begin{equation}\label{eq:cor:community recovery degree}
        s^{2}\left(\alpha + (k-1)\beta \right)>1, \; k=o(\log n)
    \end{equation} and 
    \begin{equation}\label{eq:cor:community recovery condition}
        \sqrt{\alpha}-\sqrt{\beta}>\sqrt{\frac{1}{1-(1-s)^{2}}}.
    \end{equation}
    Then, there exists an estimator $\hat{\sigma}$ such that $\hat{\sigma}=\sigma$ with probability at least $1-o(1)$ for the ground-truth community labels $\sigma=(\sigma_1,\dots, \sigma_n)$.
\end{corollary}

\begin{corollary}\label{cor:community recovery impossible}
    For $p=\frac{\alpha \log m}{m},\; q=\frac{\beta \log m}{m}$ with positive constants $\alpha, \beta$ and $s\in [0,1]$, let $G_{1}$ and $G_{2}$ be two graphs given by the correlated SBMs defined in Sec \ref{sec:model}. Suppose that 
    \begin{equation}\label{eq:cor:community recovery impossible condition}
      \sqrt{\alpha}-\sqrt{\beta}<\sqrt{\frac{1}{1-(1-s)^{2}}}\text{ and } k=n^{o(1)}.
    \end{equation} 
    Then, for any estimator $\hat{\sigma}$, we have $\lim _{n \rightarrow \infty} \P\left(\hat{\sigma}=\sigma\right)=0$.
\end{corollary}
Comparing  Corollary \ref{cor:community recovery} and Corollary \ref{cor:community recovery impossible}, we can see that Corollary \ref{cor:community recovery} requires an additional condition \eqref{eq:cor:community recovery degree}, which is not present in  Corollary \ref{cor:community recovery impossible}. This condition is the result obtained by substituting  $p=\frac{\alpha \log m}{m}, q=\frac{\beta \log m}{m}$  into  \eqref{eq:thm matching degree} and \eqref{eq:thm matching k}. Thus, these conditions are required to first match the vertices between $G_1$ and $G_2$, and then to proceed to the community recovery using $G_1\cup G_2'$.
Compared to the case of having only $G_1$, where we need $\sqrt{\alpha}-\sqrt{\beta}>\sqrt{1/s}$ for community recovery, the condition \eqref{eq:cor:community recovery condition} is relaxed since $s<1-(1-s)^2$.

As mentioned in Section \ref{sec:Introduction}, in the correlated ER model, if $nps^2 < \log n$, then the exact graph matching is impossible, but partial matching is still possible if  $nps^2 \geq \lambda_* +\epsilon$ for some constant $\lambda_*$. Similarly, in the correlated SBMs, partial matching may be possible even though the average degree $ns^{2}{\frac{p+(k-1)q}{k}}$ is less than $\log n$. If this conjecture is correct, community recovery may be possible by using partial graph matching even if \eqref{eq:cor:community recovery degree} does not hold. Given two correlated SBMs having two communities with $p=\frac{\alpha \log n}{n}$ and $q=\frac{\beta \log n}{n}$ for constants $\alpha,\beta>0$, Gaudio et al. \cite{GRS22} have recently shown that
if 
$
    \sqrt{\alpha}-\sqrt{\beta}>\sqrt{\frac{1}{1-(1-s)^{2}}}
$
and 
$
    s^{2}\left(\frac{\alpha+\beta}{2}\right)+s(1-s)\frac{(\sqrt{\alpha}-\sqrt{\beta})^{2}}{2}>1,
$
then community recovery is possible.  Moreover, community recovery is impossible if one of the two conditions is not held. This result gives us the precise information-theoretic limit for community recovery. We provide the similar conjecture in general $(p,q,k)$. 
\begin{conjecture}\label{con:recovery}
       For $p=\frac{\alpha \log m}{m},\; q=\frac{\beta \log m}{m}$ with positive constants $\alpha, \beta$ and $s\in [0,1]$,  let $G_{1}$ and $G_{2}$ be two graphs given by the correlated SBMs defined in Sec \ref{sec:model}. Suppose that 
    \begin{equation}\label{eq:con:community recovery condition}
        \sqrt{\alpha}-\sqrt{\beta}>\sqrt{\frac{1}{1-(1-s)^{2}}},
    \end{equation}
    and
    \begin{equation}\label{eq:con condition}
        s^{2}\left(\alpha + (k-1)\beta\right) + s(1-s)(\sqrt{\alpha}-\sqrt{\beta})^{2}>1
    \end{equation}
    Then, there exists an estimator $\hat{\sigma}$ such that $\hat{\sigma}=\sigma$ with probability at least $1-o(1)$.
\end{conjecture}

\subsection{Multiple correlated stochastic block models}\label{sec:multiple}
This section considers the case where there are $r$ correlated SBMs and introduces the conditions for community recovery using those $r$ correlated SBMs.

Consider the parent graph $G\sim \text{SBM}(n,p,q,k)$. We generate $r$ graphs $G_{1},G'_{2},G'_{3},\ldots ,G'_{r}$ by sampling every edge of $G$ independently with probability $s$ total $r$ times. Let $\pi^{2}_*, \ldots ,\pi^{r}_*$ be $r-1$ permutations $\pi^{j}_* :[n] \rightarrow [n]$ for $j=2,\ldots ,r$. We obtain $G_{j}$ by permuting all vertices of $G'_{j}$ with permutation $\pi^{j}_*$ for $j=2,\ldots ,r$.

\begin{corollary}\label{cor:multiple achieve}
    For $p=\frac{\alpha \log m}{m},\; q=\frac{\beta \log m}{m}$ with positive constants $\alpha, \beta$ and $s\in [0,1]$, let $G_{1},G_{2}, \dots, G_r$ be $r$-correlated SBMs. Suppose that 
    \begin{equation}\label{eq:cor:multiple matching}
      s^2 \left(\alpha + (k-1)\beta \right)>1, \; k=o(\log n) 
    \end{equation}
    and
    \begin{equation}\label{eq:cor:multiple achieve}
        \sqrt{\alpha}-\sqrt{\beta}>\sqrt{\frac{1}{1-(1-s)^{r}}}.
        \end{equation} 
    Then, there exists an estimator $\hat{\sigma}$ such that $\hat{\sigma}=\sigma$ with probability at least $1-o(1)$.
\end{corollary}

\begin{corollary}\label{cor:multiple impossible}
    For $p=\frac{\alpha \log m}{m},\; q=\frac{\beta \log m}{m}$ with positive constants $\alpha, \beta$ and $s\in [0,1]$,  let $G_{1},G_{2}, \dots, G_r$ be $r$-correlated SBMs. Suppose that 
    \begin{equation}\label{eq:cor:multiple impossible}
        \sqrt{\alpha}-\sqrt{\beta}<\sqrt{\frac{1}{1-(1-s)^{r}}} \text{ and } k=n^{o(1)}.
        \end{equation} 
   Then, for any estimator $\hat{\sigma}$, we have $\lim _{n \rightarrow \infty} \P\left(\hat{\sigma}=\sigma\right)=0$.
\end{corollary}

\section{Outline of Proof}\label{sec:outline of proof}

\subsection{Notation}
For any positive integer $n$, let $[n]=\{1,\ldots,n\}$. Recall that vertex set is $V=[n]:=\{1,2, \ldots, n\}$. Let $V_{r}:=\left\{i \in[n]: \sigma_{i}=r \right\}$ for $r \in [k]$ denote the vertex set in the $r$-th community, where $\sigma_i$ is the community the vertex $i$ belongs.

Let $\mathcal{E}:=\{\{i, j\}: i, j \in[n], i \neq j\}$ denote the set of all unordered vertex pairs. Given a community labeling $\sigma=(\sigma_1,\ldots,\sigma_n)$, we define the intra-community vertex pair set $\mathcal{E}^{+}(\sigma):=\left\{(i, j) \in \mathcal{E}: \sigma_{i} = \sigma_{j}\right\}$ and the inter-community vertex pair set $\mathcal{E}^{-}(\sigma):=\left\{(i, j) \in \mathcal{E}: \sigma_{i} \neq \sigma_{j}\right\}$. Note in particular that $\mathcal{E}^{+}(\sigma)$ and $\mathcal{E}^{-}(\sigma)$ partition $\mathcal{E}$.
Let $A$ and $B$ denote the adjacency matrices of $G_{1}$ and $G_{2}$, respectively.

For an event $\mathcal{A}$, let $\textbf{1}(\mathcal{A})$ be the indicator random variable.

\subsection{Achievability for exact graph matching}
We will first provide the proof overview of Theorem \ref{thm:matching achievability}, the achievability of exact graph matching, by generalizing the proof techniques of  \cite{RS21} for the two communities with $p=\Theta(q)=\Theta(\log n/n)$ to general $(p,q,k)$.
Consider the estimator that maximizes the alignment between the two graphs:
\begin{equation}\label{eq:achieve estimator}
    \hat{\pi} \in \argmax_{\pi \in S_{n}} \Sigma_{(i,j)\in \mathcal{E}} A_{i,j}B_{\pi(i),\pi(j)},
\end{equation}
where $S_{n}$ is the set of all permutations over $n$ vertices. This estimator $\hat{\pi}$ maximizes the number of common edges between $G_{1}$ and $G_{2}$. Let us define the lifted permutation.
\begin{definition}[Definition 2.1 in \cite{RS21}] \label{eq:lifted permutation}
Consider a permutation $\pi \in \mathcal{S}_n$ that operates on the vertices. We define a corresponding lifted permutation $\tau: \mathcal{E} \rightarrow \mathcal{E}$ on pairs of vertices, where $\tau((i, j)) := (\pi(i), \pi(j))$. To simplify the notation, we use $\tau=\ell(\pi)$ as shorthand. Consequently, we also denote $\tau_*=\ell(\pi_*)$ and $\hat{\tau}=\ell(\hat{\pi})$.

\end{definition}

\begin{definition}\label{def:balanced community}
    For $\epsilon>0$, define an event
    \begin{equation}
        \mathcal{F}_{\epsilon}:=\left\{\left(1-\frac{\epsilon}{2} \right)\frac{n}{k} \leq |V_{r}| \leq \left(1+\frac{\epsilon}{2} \right)\frac{n}{k} \text{ for } r\in[k] \right\}.
    \end{equation} 
\end{definition}
When the probability that each vertex belongs to each community is uniform, $\mathcal{F}_{\epsilon}$ holds with probability $1-o(1)$, where $k^2 \log k = o(n)$. 

Note that
\begin{equation}
    \begin{aligned}
      \mathbb{P}\left(\hat{\pi} \neq \pi_*\right)&=\mathbb{P}\left(\hat{\tau} \neq \tau_*\right)=\mathbb{E}\left[\mathbb{P}\left(\hat{\tau} \neq \tau_* \mid \sigma, \tau_*\right)\right] \\
      &\leq \mathbb{E}\left[\mathbb{P}\left(\hat{\tau} \neq \tau_* \mid \sigma, \tau_*\right) \mathbf{1}\left(\mathcal{F}_\epsilon\right)\right]+\mathbb{P}\left(\mathcal{F}_\epsilon^c\right).
    \end{aligned}
\end{equation}
To show $\mathbb{P}\left(\hat{\tau} \neq \tau_* \mid \sigma, \tau_*\right) \mathbf{1}\left(\mathcal{F}_\epsilon\right)$ is close to 0 for any wrong permutation $\hat{\tau} \neq \tau_{*}$, we will use the probability generating function (PGF)  as defined in \cite{RS21}:
\begin{equation}\label{eq:definition pgf}
    \Phi^\tau(\theta, \omega, \zeta):=\mathbb{E}\left[\theta^{X(\tau)} \omega^{Y^+(\tau)} \zeta^{Y^{-}(\tau)} \big| \sigma, \tau_*\right],
\end{equation}
where 
\begin{equation}
    \begin{aligned}
      X(\tau):&=\sum_{e \in \mathcal{E}} A_e B_{\tau_*(e)}-\sum_{e \in \mathcal{E}} A_e B_{\tau(e)}\\
     &=\sum_{e \in \mathcal{E}: \tau(e) \neq \tau_*(e)}\left(A_e B_{\tau_*(e)}-A_e B_{\tau(e)}\right) ,\\
        Y^{+}(\tau) & :=\sum_{e \in \mathcal{E}^{+}({\sigma}): \tau(e) \neq \tau_*(e)} A_e B_{\tau_*(e)}, \\
Y^{-}(\tau) & :=\sum_{e \in \mathcal{E}^{-}({\sigma}): \tau(e) \neq \tau_*(e)} A_e B_{\tau_*(e) } .
    \end{aligned}
\end{equation}
Note that $X(\tau_*)=0$. $X(\hat{\tau})$ must be less than or equal to 0 for $\hat{\tau}\neq \tau_*$ for the optimization \eqref{eq:achieve estimator} to output $\hat{\tau}= \tau_*$.
We will use this fact and  an upper bound for $\Phi^\tau(\theta, \omega, \zeta)  $ to bound  $\mathbb{P}\left(\hat{\tau} \neq \tau_* \mid \sigma, \tau_*\right) \mathbf{1}\left(\mathcal{F}_\epsilon\right)$. 

\begin{lemma}
\label{lem:PGF}
Suppose that $\sigma$ and $\tau_*=\ell(\pi_*)$ are given. Fix $\pi \in \mathcal{S}_n$ and let $\tau=\ell(\pi)$. For any constants $\epsilon \in(0,1)$, $p,q=o(1)$, $\theta=(\sqrt{p}\vee\sqrt{q})$, and $1 \leq \omega, \zeta \leq 3$,  for sufficiently large $n$, we have that
\begin{equation}
    \Phi^\tau(\theta, \omega, \zeta) \leq \exp \left(-(1-\epsilon) s^2\left(p M^{+}(\tau)+q M^{-}(\tau)\right)\right),
\end{equation}
where
\begin{equation}
    \begin{aligned}
M^{+}(\tau) & :=\left|\left\{e \in \mathcal{E}^{+}({\sigma}): \tau(e) \neq \tau_*(e)\right\}\right|, \\
M^{-}(\tau) & :=\left|\left\{e \in \mathcal{E}^{-}({\sigma}): \tau(e) \neq \tau_*(e)\right\}\right|. 
\end{aligned}
\end{equation}
\end{lemma}
Racz and Sridhar \cite{RS21} only considered the case $p,q=\Theta\left(\frac{\log n}{n}\right)$ with two communities. We extended the result and the corresponding lemma (Lemma 2.3 in \cite{RS21}) for more general $p$ and $q$ by choosing an appropriate value of $\theta$. As long as $p\theta^{-1}, q\theta^{-1}=o(1)$, Lemma \ref{lem:PGF} can be established, so we assume that $p, q=o(1)$ and $\theta=(\sqrt{p} \vee \sqrt{q})$.

Recall that we have the vertex sets of $k$ communities, $V_{1},\ldots,V_{k}$.
Define the set of all lifted permutations for which there are $x_{r}$ incorrectly matched vertices in the $r$-th community for $r\in[k]$:
\begin{equation}
    S_{x_1, \ldots x_k}:=\left\{\ell(\pi):\big|\left\{ i \in V_r: \pi(i) \neq \pi_*(i)\right\}\big|=x_r, r \in[k]\right\}.
\end{equation}
For $\hat{\tau}\neq \tau_*$,  we have
\begin{equation}\label{eq:bound}
\begin{aligned}
    &\mathbb{P}\left(\hat{\tau} \neq \tau_* \mid \sigma, \tau_*\right) \mathbf{1}\left(\mathcal{F}_\epsilon\right)\\
    &=\sum_{x_1+\ldots+x_k \geq 2}\P( \hat{\tau} \in S_{x_1,\ldots,x_k} \mid \sigma, \tau_*) \mathbf{1}\left(\mathcal{F}_\epsilon\right),
\end{aligned}
\end{equation}
where $x_1+\ldots+x_k\geq 2$ is from the fact that it cannot happen that only one node is mismatched.

Let $\tau \in S_{x_1,\ldots,x_k}$. For the output $\hat{\pi}$ of \eqref{eq:achieve estimator} and its lifted permutation $\hat{\tau}=\ell(\hat{\pi})$,  we have
\begin{equation}\label{eq:error}
    \begin{aligned}
\mathbb{P}\left(\hat{\tau}=\tau \mid \sigma, \tau_*\right) & \leq  \mathbb{P}\left(X(\tau) \leq 0 \mid \sigma, \tau_*\right)\\
&\stackrel{(a)}{=}\mathbb{P}\left(\theta^{X(\tau) } \geq 1 \mid \sigma, \tau_*\right) \\
& \stackrel{(b)}{\leq} \Phi^\tau(\theta, 1,1) \\
&\leq \exp \left(-(1-\epsilon) s^2\left(p M^{+}(\tau)+q M^{-}(\tau)\right)\right).
\end{aligned}
\end{equation}
The equality $(a)$ holds by $0<\theta <1$, the inequality $(b)$ holds by Markov's inequality and \eqref{eq:definition pgf}, and the last inequality holds from Lemma \ref{lem:PGF}. Therefore, if $M^{+}(\tau)$ and $M^{-}(\tau)$ are large enough, we can make $\mathbb{P}\left(\hat{\tau}=\tau \mid \sigma, \tau_*\right)$  small enough. Let us define two sets $ E_{\text {tr }}^{+}$ and $ E_{\text {tr }}^{-}$ as follows:
\begin{equation}
    \begin{aligned}
& E_{\text {tr }}^{+}:=\left\{(u, v) \in \mathcal{E}^{+}(\sigma): \pi(u)=\pi_*(v), \pi(v)=\pi_*(u)\right\}, \\
& E_{t r}^{-}:=\left\{(u, v) \in \mathcal{E}^{-}(\sigma): \pi(u)=\pi_*(v), \pi(v)=\pi_*(u)\right\} .
\end{aligned}
\end{equation}
We express $M^{+}(\tau)$ and $M^{-}(\tau)$ using $ E_{\text {tr }}^{+}$ and $ E_{\text {tr }}^{-}$, and the result can be represented by the following lemma:
\begin{lemma}
\label{lem:M}
    Suppose that $\sigma$ and $\pi_*$ are given. Fix $\pi \in \mathcal{S}_n$ and let $\tau=\ell(\pi)$. Consider $x_1, \ldots, x_k$ such that $\tau \in S_{x_1, \ldots, x_k}$. Then we have that
\begin{equation}
\begin{aligned}
& M^{+}(\tau)=\sum_{r=1}^k {x_r \choose 2}+x_r\left(\left|V_r\right|-x_r\right)-\left|E_{t r}^{+}\right| \\
& M^{-}(\tau)=\sum_{r=1}^k x_r\left(n-\left|V_r\right|\right)-\sum_{r<r^{\prime}} x_r x_{r^{\prime}}-\left|E_{t r}^{-}\right|.
\end{aligned}
\end{equation}
Furthermore, we have that $\left|E_{t r}^{+}\right|,\left|E_{t r}^{-}\right| \leq\left(x_1+\ldots+x_k\right) / 2$.
\end{lemma}

From the definitions of $M^{+}(\tau)$ and $M^{-}(\tau)$, the above lemma can be easily proved. By Lemma \ref{lem:M}, on the event $\mathcal{F}_{\epsilon}$, if $x_{r} \leq \frac{\epsilon}{2} |V_{r}|$ for all $r\in[k]$, we can show
\begin{equation}\label{eq:M easy}
    \begin{aligned}
        & M^{+}(\tau) \geq (1-\epsilon)\frac{n}{k}(x_{1}+\ldots,x_{k}),\\
        & M^{-}(\tau) \geq (1-\epsilon)\frac{(k-1)n}{k}(x_{1}+\ldots,x_{k}).
    \end{aligned}
\end{equation}
Plugging \eqref{eq:M easy} back into \eqref{eq:error} and taking a union bound over all $\tau\in S_{x_1,\ldots,x_k}$, we get 
\begin{equation}\label{eq:prob bound1}
		\mathbb{P}\left(\hat{\tau} \in S_{x_{1},\ldots,x_{k}} \mid \sigma, \tau_{*}\right) 1\left(\mathcal{F}_{\epsilon}\right) \leq n^{-\epsilon\left(x_{1}+\ldots+x_{k}\right)} .
\end{equation}
In general case where $x_{r} \geq \frac{\epsilon}{2} |V_{r}|$ for some $r\in[k]$, on the other hand, we have 
\begin{equation}\label{eq:M general}
    \begin{aligned}
        & M^{+}(\tau) \geq (1-\epsilon)\frac{n}{2k}(x_{1}+\ldots,x_{k}),\\
        & M^{-}(\tau) \geq (1-\epsilon)\frac{(k-1)n}{2k}(x_{1}+\ldots,x_{k}).
    \end{aligned}
\end{equation}
In this case,  $M^{+}(\tau)$ and  $M^{-}(\tau)$ are not large enough, and thus we need a different approach than \eqref{eq:error}. If $x_{r} \geq \frac{\epsilon}{2} |V_{i}|$ for some $r\in[k]$, we will show below lemma.
\begin{lemma}
\label{lem:prob bound 2}
    For $p,q = o(1)$, $s \in[0,1]$ and any $\epsilon \in (0,1)$ such that  $ nps^{2}, nqs^{2}=\omega(1)$ and  $ns^{2}(p+(k-1)q) / k>(1+\epsilon)(1-\epsilon)^{-2} \log n$, there exists a positive constant $C$ with the properties below.
    Given $\sigma$, let $x_{1},\ldots,x_{k}$ be such that $x_{r} \geq \frac{\epsilon}{2}\left|V_{r}\right|$ for some $r\in[k]$. For sufficiently large $n$, we have that
	\begin{equation}\label{eq:prob bound2}
		\mathbb{P}\left(\hat{\tau} \in S_{x_{1},\ldots, x_{k}} \mid \sigma, \tau_{*}\right) 	\mathbf{1}\left(\mathcal{F}_{\epsilon}\right) \leq n^{-\delta\left(x_{1}+\ldots+x_{k}\right)},
	\end{equation}
 where $\delta=\left(\epsilon/2 \wedge \frac{\epsilon^{2}(1-\epsilon)nps^{2}}{100 k \log n} \wedge \frac{\epsilon^{2}(k-1)nqs^{2}}{100 k \log n}\right)$ if $k \leq C(\sqrt{nps^{2}}\wedge nqs^{2})$.
\end{lemma}
The proof for this lemma will be described in detail in Appendix \S\ref{app:A}.
By combining \eqref{eq:prob bound1} and \eqref{eq:prob bound2}, we have
\begin{equation}\label{eq:final bound}
   \mathbb{P}\left(\hat{\tau} \in S_{x_{1},\ldots,x_{k}} \mid \sigma, \tau_{*}\right) 1\left(\mathcal{F}_{\epsilon}\right)\leq n^{-\delta(x_{1}+\ldots+x_{k})},
\end{equation}
where $\delta=\left(\epsilon/2 \wedge \frac{\epsilon^{2}(1-\epsilon)nps^{2}}{100 k \log n} \wedge \frac{\epsilon^{2}(k-1)nqs^{2}}{100 k \log n}\right)$.
Plugging \eqref{eq:final bound} into \eqref{eq:bound} and taking union bound over $x_{1}+\ldots+x_{k} \geq 2$, the proof is complete. The detailed proof is in Appendix \S\ref{app:A}.

\subsection{Impossibility for exact graph matching}
 We will consider the maximum a posteriori (MAP) estimator to prove Theorem \ref{thm:matching impossible}. The MAP estimator is given by
\begin{equation}
\hat{\pi}_{\mathrm{MAP}} \in \underset{\pi \in \mathcal{S}_{n}}{\arg \max } \;\;\mathbb{P}\left(\pi_{*}=\pi \mid A, B, \sigma\right).
\end{equation}
To prove the impossibility of graph matching, many previous works, including \cite{RS21} and \cite{CSKM16}, have used the fact that the exact graph matching is impossible if there exists an isolated vertex. Thus, given the ground-truth community labels, we found the conditions for at least one community to have isolated vertices. These conditions are \eqref{eq:thm:matching impossible degree} and \eqref{eq:thm:matching impossible k}.

\subsection{Community recovery}
For the stochastic block model with $n=km$ vertices partitioned into $k$ equal-sized communities, in the regime  $p=\frac{\alpha \log m}{m}$, $q=\frac{\beta \log m}{m}$ for $\alpha>\beta>0$, Agarwal et al. \cite{ABKK17}  showed that the exact community recovery is possible if $\sqrt{\alpha}-\sqrt{\beta}>1$ for $k=o(\log n)$ and it is impossible with high probability if  $\sqrt{\alpha}-\sqrt{\beta}<1$, where $k=n^{o(1)}$. 
Substituting the same condition to \eqref{eq:thm matching k} gives $k \leq C\cdot s \log n$ and substituting it to \eqref{eq:thm matching degree} gives $s^2(\alpha+(k-1)\beta)>1$.

Thus, we can apply Theorem \ref{thm:matching achievability} under the condition \eqref{eq:cor:community recovery degree}, and show that the exact matching of vertices between $G_1$ and $G_2$ is achievable. Then, by combining $G_1$ and $G'_2$ (after permutation the vertices of $G_2$ by $\pi^{-1}_*$), we get $G_1\cup G'_2$ with intra-community edge density $p'=\frac{\alpha (1-(1-s)^2)\log m}{m}$ and inter-community edge density $q'=\frac{\beta (1-(1-s)^2)\log m}{m}$. Thus, by applying the result from  \cite{ABKK17}, we can prove Corollary \ref{cor:community recovery} and \ref{cor:community recovery impossible}. Similarly, Corollary \ref{cor:multiple achieve} and \ref{cor:multiple impossible} can also be proved.

\section{Conclusion and Open Problems}

In this paper, we considered graph matching on correlated stochastic block models with applications for improving graph clustering. 
In particular, we extended the result of \cite{RS21} to the general $(p,q,k)$ regime. 
For general $p>q$ with $p,q=o(1)$ and $nps^2,nqs^2=\omega(1)$, we showed that the exact matching is possible if the average degree of vertices satisfies $ns^{2}\left(\frac{p+(k-1)q}{k}\right)>(1+\epsilon)\log n$ for any $\epsilon>0$ with an extra condition on $k$ in \eqref{eq:thm matching k}, and impossible if $ns^{2}\left(\frac{p+(k-1)q}{k}\right)<(1-\epsilon)\log n$ for $k=O(n^{t})$ for some $t>0$. By using this result, we also derived the regime of $(p,q,s)$ where having a correlated graph as side information can bring information advantage in recovering the community structure of the SBM graph. In particular, compared to the case of having a single SBM graph, where the exact community detection is achievable only if $\sqrt{\alpha}-\sqrt{\beta}>\sqrt{1/s}$ for the equal-sized communities of density $p=s\alpha\log m/m$ and $q=s\beta\log m/m$ with $\alpha>\beta>0$, the case of having two correlated SBMs can achieve the exact community detection if $\sqrt{\alpha}-\sqrt{\beta}>\sqrt{\frac{1}{1-(1-s)^{2}}}$ when the condition \eqref{eq:cor:community recovery degree} is satisfied. Our work has left several interesting open questions:


\begin{itemize}
    
    \item \textbf{Information-theoretic limit for community recovery in correlated SBMs:}  The achievability result for exact community recovery in Corollary \ref{cor:community recovery} is derived assuming that the exact matching of the two graphs is available through the condition in \eqref{eq:cor:community recovery degree} and Theorem \ref{thm:matching achievability}. However, as conjectured in Conjecture \ref{con:recovery}, the exact matching may not be necessary but partial matching can be sufficient for guaranteeing the exact community recovery with \eqref{eq:con:community recovery condition}.   
    Gaudio et al. \cite{GRS22} found the precise information-theoretic limit for community recovery on the correlated SBMs with two communities, and Conjecture \ref{con:recovery} aims to generalize this result for $k\geq2$. Proving this conjecture, or deriving the information-theoretic limit for community recovery in correlated SBMs with general $k$ is an interesting open direction.
    
    \item \textbf{Finding efficient algorithms for exact graph matching:}
    The achievability result for exact graph matching in Theorem \ref{thm:matching achievability} considers the estimator that searches over all the possible permutations to maximize the alignment between the two given graphs. However, the complexity of solving such an optimization scales as $\Theta(n!)$, which can be a prohibitive complexity with large $n$. Thus, the remaining question is how to design an efficient graph matching algorithm for the correlated SBMs. 
   In the correlated ER model, there have been many efficient matching algorithms including the seeded matching \citep{MX20,YG13}, noisy seeded matching \citep{YXL21}, and seedless graph matching algorithms \citep{BCL18+,FMWX19a,FMWX19b,MRT21a,MRT21b,MWXY22}. 
    The main question is how to extend and generalize such efficient algorithms to correlated SBMs. Recently, an efficient polynomial-time algorithm for exact matching of correlated SBMs was proposed in \citep{yang2023efficient}, under the assumption that the community structure is recovered in both graphs. It is still open though how to modify such an algorithm for the case of unknown community structure.

\end{itemize}

\bibliographystyle{IEEEtran}
\bibliography{arxivfull_rev2}

\appendices
\section{Proof of Theorem \ref{thm:matching achievability}}\label{app:A}

Recall the event 
    \begin{equation}
        \mathcal{F}_{\epsilon}:=\left\{\left(1-\frac{\epsilon}{2} \right)\frac{n}{k} \leq |V_{r}| \leq \left(1+\frac{\epsilon}{2} \right)\frac{n}{k} \text{ for } r\in[k] \right\}.
    \end{equation} 

\begin{lemma}\label{lem:balanced community}
Suppose that the community label of each vertex is uniformly distributed over $[k]$. Then, we have
    \begin{equation}
    \P\{ \mathcal{F}_{\epsilon} \}=1-o(1),
\end{equation}
    where $k^{2} \log k = o(n)$.
\end{lemma}
\begin{proof}
     For any $r\in [k]$, by Hoeffding's inequality we have
\begin{equation}
    \P\left\{  \left(1-\frac{\epsilon}{2} \right)\frac{n}{k} \leq |V_{r}| \leq \left(1+\frac{\epsilon}{2} \right)\frac{n}{k} \right\} \geq 1-\exp \left(-C \frac{n}{k^{2}} \right),
\end{equation}
where $C$ is a constant depending on $\epsilon$. Taking a union bound over $[k]$, we get 
  \begin{equation*}
    \P\{ \mathcal{F}_{\epsilon} \}\geq 1-k \exp \left(-C \frac{n}{k^{2}} \right) = 1-\exp \left(\log k - C \frac{n}{k^2}\right).
\end{equation*}
Therefore, if $k^{2} \log k = o(n)$ then $\P\{ \mathcal{F}_{\epsilon} \} \rightarrow 1$. Thus, the proof is complete.
\end{proof}
We next address lemmas, which are extensions of Lemma 2.6-2.9 of \citep{RS21}, respectively, that will be used to prove Theorem \ref{thm:matching achievability}.
\begin{lemma}
\label{lem:M bound 1}
	Fix $\epsilon>0$.  Let $x_{r}$ be such that $x_{r} \leq \frac{\epsilon}{2}\left|V_{r}\right|$ for all $r\in [k]$. Suppose that event $\mathcal{F}_{\epsilon}$ holds. Given the ground truth permutation $\pi_*$, let $\tau$  be a lifted permutation such that $\tau \in S_{x_{1},\ldots,x_{k}}$. For sufficiently large $n$, we obtain that		
	\begin{align}\label{eq:M+ bound1}
		&M^{+}(\tau) \geq(1-\epsilon) \frac{n}{k}\left(x_{1}+\ldots+x_{k}\right) \\
		\label{eq:M- bound1}&M^{-}(\tau) \geq(1-\epsilon) \frac{(k-1)n}{k}\left(x_{1}+\ldots+x_{k}\right).
	\end{align}
\end{lemma}
\begin{proof}
First, we will show \eqref{eq:M+ bound1}. By Lemma \ref{lem:M}, we have
\begin{equation}
	\begin{aligned}
		M^{+}(\tau) &=  \sum^{k}_{r=1}{x_{r} \choose 2}+x_{r}\left(\left|V_{r}\right|-x_{r}\right)-\left|E_{t r}^{+}\right|\\
        & \stackrel{(a)}{\geq}\sum^{k}_{r=1}x_{r}\left(\left|V_{r}\right|-x_{r}\right)-\frac{x_{1}+\ldots+x_{k}}{2}\\
        &\stackrel{(b)}{\geq}\sum^{k}_{r=1}\left(1-\frac{\epsilon}{2}\right)x_{r}|V_{r}|-\frac{x_{1}+\ldots+x_{k}}{2} \\
		& \stackrel{(c)}{\geq}\left(\left(1-\frac{\epsilon}{2}\right)^{2} 	\frac{n}{k}-\frac{1}{2}\right)(x_{1}+\ldots+x_{k})\\
        &\stackrel{(d)}{\geq}(1-\epsilon) 	\frac{n}{k}\left(x_{1}+\ldots+x_{k}\right).
	\end{aligned}
\end{equation}
The inequality $(a)$ holds from ${x_{r} \choose 2} \geq 0$ and $\left|E_{t r}^{+}\right|\leq \left(x_{1}+\ldots+x_{k}\right) / 2$, the inequality $(b)$ holds by $x_{r} \leq \frac{\epsilon}{2}\left|V_{r}\right|$ for all $r \in [k]$, the inequality $(c)$ holds by $\left|V_{r}\right| \geq$ $(1-\epsilon / 2) n / k$ on the event $\mathcal{F}_{\epsilon}$, and the last inequality $(d)$ holds since $(1-\epsilon / 2)^{2}>1-\epsilon$. Similarly, we have
\begin{equation}
    \begin{aligned}
	M^{-}(\tau) &=\sum^{k}_{r=1}x_{r}\left(n-|V_{r}|\right)-\sum_{r<r'}x_{r}x_{r'}-\left|E_{t r}^{-}\right|\\
 & \stackrel{(a)}{\geq} \sum^{k}_{r=1}x_{r}\left(n-|V_{r}|\right)-\sum_{r<r'}x_{r}x_{r'}-\frac{x_{1}+\ldots+x_{k}}{2}\\
 &\stackrel{(b)}{=}\sum^{k}_{r=1}x_{r}\left(\sum_{r'\neq r}\left(V_{r'}-\frac{x_{r'}}{2}-\frac{1}{2(k-1)}\right)\right)\\
	& \stackrel{(c)}{\geq} \sum^{k}_{r=1}x_{r}\left(\sum_{r'\neq r}\left(1-\frac{\epsilon}{2}\right)V_{r'}\right)\\
 &\stackrel{(d)}{\geq}\left(1-\frac{\epsilon}{2}\right)^{2} \frac{(k-1)n}{k}\left(x_{1}+\ldots+x_{k}\right) \\
	&\geq(1-\epsilon) \frac{(k-1)n}{k}\left(x_{1}+\ldots+x_{k}\right).
\end{aligned}
\end{equation}
The inequality $(a)$ holds from $\left|E_{t r}^{-}\right| \leq \left(x_{1}+\ldots+x_{k}\right) / 2$, the equality $(b)$ holds by $\sum^k_{i=1}|V_i| = n$, the inequality $(c)$ holds since $x_{r'}+1 \leq \epsilon\left|V_{r'}\right|$ for all $r' \in [k]$, the inequality $(d)$ holds since $\left|V_{r}\right| \geq(1-\epsilon / 2) n / k$ for all $r \in [k]$ on the event $F_{\epsilon}$ and the last inequality holds since $(1-\epsilon / 2)^{2}>1-\epsilon$. Thus, the proof is complete.
\end{proof}

We will use Lemma \ref{lem:M bound 1} to prove Lemma \ref{lem:prob bound 1}.

\begin{lemma}
\label{lem:prob bound 1}
	Fix $p,q=o(1)$, $ s \in[0,1]$ and $\epsilon \in(0,1)$ such that $ns^{2}(p+(k-1)q) / k>(1+\epsilon)(1-\epsilon)^{-2}\log n$. Given $\sigma$, let $x_{1},\ldots,x_{k}$ be such that $x_{r} \leq \frac{\epsilon}{2}\left|V_{r}\right|$ for all $r\in[k]$. For sufficiently large $n$, we obtain that
	\begin{equation}
		\mathbb{P}\left(\hat{\tau} \in S_{x_{1},\ldots,x_{k}} \mid \sigma, \tau_{*}\right) 1\left(\mathcal{F}_{\epsilon}\right) \leq n^{-\epsilon\left(x_{1}+\ldots+x_{k}\right)} .
	\end{equation}
\end{lemma}
\begin{proof}
    Let $\theta=\sqrt{p} \vee \sqrt{q}$ and $\tau \in S_{x_{1},\ldots,x_{k}}$. By \eqref{eq:error}, we have
$$
\begin{aligned}
	\mathbb{P}\left(\hat{\tau}=\tau | \sigma, \tau_* \right)\leq \exp \left(-(1-\epsilon) s^{2}\left(p M^{+}(\tau)+q M^{-}(\tau)\right) \right).
\end{aligned}
$$
Note that we can bound $|S_{x_1,\ldots,x_k}|$ as below.
\begin{equation}\label{eq:S bound}
\begin{aligned}
    \left|S_{x_{1},\ldots,x_{k}}\right| &\stackrel{(a)}{\leq} {n \choose x_{1}+\ldots+x_{k}}  \left(x_{1}+\ldots+x_{k}\right)!\\
     &=\frac{n !}{\left(n-x_{1}-\ldots-x_{k}\right) !} \leq n^{x_{1}+\ldots+x_{k}} .
\end{aligned}
\end{equation}
The inequality $(a)$ holds since the first term is the number of ways to choose mismatched vertices, and the second term is the number of ways to permute these.
Taking union bound, we can obtain
\begin{equation}\label{eq:exp bound1}
\begin{aligned}
	\mathbb{P}\left(\hat{\tau} \in S_{x_{1},\ldots,x_{k}} \mid \sigma, \tau_{*}\right) & \leq\left|S_{x_{1},\ldots,x_{k}}\right| \max _{\tau } \mathbb{P}\left(\hat{\tau}=\tau \mid \sigma, \tau_{*}\right) \\
	& \leq \max _{\tau \in S_{x_1,\ldots,x_k}}\left\{ \exp \left(\left(x_{1}+\ldots+x_{k}\right) \log n \right.\right.\\
 &\left.\left.\;\;-s^{2}(1-\epsilon) \left(p M^{+}(\tau)+q M^{-}(\tau)\right) \right)\right\} .
\end{aligned}
\end{equation}

On the event $\mathcal{F}_{\epsilon}$, by Lemma \ref{lem:M bound 1} and assumption that $ns^{2}\frac{(p+(k-1)q)}{k }>(1+\epsilon)(1-\epsilon)^{-2}\log n $, we have
\begin{equation*}
\begin{aligned}
&\left(x_{1}+\ldots+x_{k}\right) \log n-(1-\epsilon) s^{2}\left(p M^{+}(\tau)+q M^{-}(\tau)\right)\\
 &\leq   	\left\{1-(1-\epsilon)^{2} ns^{2}\frac{(p+(k-1)q)}{k \log n}  \right\}\left(x_{1}+\ldots+x_{k}\right) \log n \\
     &\leq-\epsilon\left(x_{1}+\ldots+x_{k}\right) \log n,
\end{aligned}
\end{equation*}
and 
$$
\exp \left(-\epsilon\left(x_{1}+\ldots+x_{k}\right) \log n\right)=n^{-\epsilon\left(x_{1}+\ldots+x_{k}\right)}.
$$
Thus, the proof is complete.
\end{proof}

\begin{lemma}
\label{lem:M bound 2}
Fix $\epsilon>0$. Suppose that event $\mathcal{F}_{\epsilon}$ holds. Given the ground truth permutation $\pi_{*}$, let $x_{1},\ldots,x_{k}$  be such that $\tau \in S_{x_{1},\ldots,x_{k}}$. For sufficiently large $n$, we obtain that
\begin{align}\label{eq:M+ bound2}
	&M^{+}(\tau) \geq(1-\epsilon) \frac{n}{2k}\left(x_{1}+\ldots+x_{k}\right), \\
	\label{eq:M- bound2}&M^{-}(\tau) \geq(1-\epsilon) \frac{(k-1)n}{2k}\left(x_{1}+\ldots+x_{k}\right).
\end{align}

\end{lemma}
\begin{proof}
     Similar to the proof of Lemma \ref{lem:M bound 1}, we have
$$
\begin{aligned}
	M^{+}(\tau) & \geq \sum^{k}_{r=1}{x_{r} \choose 2}+x_{r}\left(\left|V_{r}\right|-x_{r}\right)-\frac{x_{1}+\ldots+x_{k}}{2}\\
	&=\sum^{k}_{r=1}x_{r}\left(|V_{r}|-\frac{x_{r}+2}{2}\right) \\
	& \stackrel{(a)}{\geq} \sum^{k}_{r=1}\frac{1}{2}x_{r}(|V_{r}|-2) \geq (1-\epsilon) \frac{n}{2k}\left(x_{1}+\ldots+x_{r}\right).
\end{aligned}
$$
The inequality $(a)$ holds since $x_{r} \leq\left|V_{r}\right|$ for all $r\in[k]$ and the last inequality holds since $\left|V_{r}\right|-2 \geq(1-\epsilon / 2)\left|V_{r}\right|$ as well as $\left|V_{r}\right| \geq(1-\epsilon / 2) n / k$ for all $r\in[k]$. Similarly, \eqref{eq:M- bound2} can also be proven.
\end{proof}

We will use Lemma \ref{lem:M bound 2} to prove Lemma \ref{lem:prob bound 2}.

\begin{proof}[Proof of Lemma \ref{lem:prob bound 2}]
For notational simplicity, we will use $X,Y^+,Y^-$ instead of $X(\tau),Y^+(\tau),Y^-(\tau)$.  

For any $z^{+}$and $z^{-}$, we obtain that
\begin{align}
	&\mathbb{P}\left(\hat{\tau} \in S_{x_{1},\ldots,x_{k}} \mid \sigma, \tau_{*}\right) \leq  \mathbb{P}\left(\exists \tau \in S_{x_{1},\ldots,x_{k}}: X \leq 0 \mid \sigma, \tau_{*}\right) \\
	&\label{eq:case1} \leq \mathbb{P}\left(\exists \tau \in S_{x_{1},\ldots,x_{k}}: X \leq 0, Y^{+}\geq z^{+}, Y^{-} \geq z^{-} \mid \sigma, \tau_{*}\right) \\
	&\label{eq:case2}\quad+\mathbb{P}\left(\exists \tau \in S_{x_{1},\ldots,x_{k}}: Y^{+} \leq z^{+} \mid \sigma, \tau_{*}\right) \\
	&\label{eq:case3}\quad+\mathbb{P}\left(\exists \tau \in S_{x_{1},\ldots,x_{k}}: Y^{-} \leq z^{-} \mid \sigma, \tau_{*}\right) .
\end{align}
First, we will show that \eqref{eq:case1} is bounded by $n^{-\epsilon\left(x_{1}+\ldots+x_{k}\right)}$. We have 
\begin{equation}\label{eq:bound first}
    \begin{aligned}
&\mathbb{P}(X  \left.\leq 0, Y^{+} \geq z^{+}, Y^{-} \geq z^{-} \mid \sigma, \tau_*\right) \\
& =\sum_{k \leq 0} \sum_{k^{+} \geq z^{+}} \sum_{k^{-} \geq z^{-}} \mathbb{P}\left(\left(X, Y^{+}, Y^{-}\right)=\left(k, k^{+}, k^{-}\right) \mid \sigma, \tau_*\right) \\
& \stackrel{(a)}{\leq} \sum_{k=-\infty}^{\infty} \sum_{k^{+}=-\infty}^{\infty} \sum_{k^{-}=-\infty}^{\infty} \theta^k \omega^{k^{+}-z^{+}} \zeta^{k^{-}-z^{-}} \\
& \qquad  \qquad \qquad \quad \times \mathbb{P}\left(\left(X, Y^{+}, Y^{-}\right)=\left(k, k^{+}, k^{-}\right) \mid \sigma, \tau_*\right) \\
& =\omega^{-z^{+}} \zeta^{-z^{-}} \Phi^\tau(\theta, \omega, \zeta) .
\end{aligned}
\end{equation}

The inequality $(a)$ holds for $0<\theta<1$ and $\omega,\zeta  \geq 1$. 
Let $\theta = \sqrt{p} \vee \sqrt{q}$, $\omega=e$, $\zeta  =e $ and 
\begin{equation*}
	z^{+}:=(1-\epsilon) s^{2} p M^{+}(\tau)  \text { and }  z^{-}:=(1-\epsilon) s^{2} q M^{-}(\tau).
\end{equation*}
By Lemma \ref{lem:M bound 2}, Lemma \ref{lem:PGF}, \eqref{eq:bound first} and  $|S_{x_1,\ldots,x_k}| \leq n^{x_1+\ldots+x_k}$ (result from \eqref{eq:S bound}), \eqref{eq:case1} is bounded by
$$
\begin{aligned}
&\max _{\tau } \exp \left(\sum^k_{i=1}x_i \log n-2(1-\epsilon) s^2\left(p M^{+}(\tau)+qM^{-}(\tau)\right)\right)\\
&\leq \exp \left( \left(1-(1-\epsilon)^2 ns^2 \frac{(p+(k-q))}{ \log n}\right) \sum^k_{i=1}x_i  \log n \right)\\
&\leq \exp\left(-\epsilon\left(x_1+\ldots+x_k\right) \log n \right)=n^{-\epsilon\left(x_{1}+\ldots+x_{k}\right)}.
\end{aligned}
$$
The last inequality holds by assumption $ns^{2}(p+(k-1)q) / k>(1+\epsilon)(1-\epsilon)^{-2} \log n$. Thus, \eqref{eq:case1} is bounded by $n^{-\epsilon\left(x_{1}+\ldots+x_{k}\right)}$.

By using the similar way as the proof of Lemma 2.9 in \citep{RS21}, we can show that \eqref{eq:case2} is bounded by $\exp \left(-\frac{\epsilon^{2}(1-\epsilon)}{50k}\left(x_{1}+\ldots+x_{k}\right)  np s^{2}\right)$, if  $k\leq C\sqrt{nps^{2}}$ for a sufficiently small $C$. Similarly, we can show that \eqref{eq:case3} is bounded by $\exp \left(-\frac{(k-1)\epsilon^{2}(1-\epsilon)}{50k}\left(x_{1}+\ldots+x_{k}\right) nqs^{2}\right)$, if $k\leq Cnqs^{2}$.  Let 
$$
\delta_{0}:=\epsilon \wedge \frac{\epsilon^{2}(1-\epsilon)nps^{2}}{50 k \log n} \wedge \frac{\epsilon^{2}(k-1)nqs^{2}}{50 k \log n}.
$$
Then, the three terms \eqref{eq:case1}, \eqref{eq:case2}), and \eqref{eq:case3} are bounded by $n^{-\delta_{0}\left(x_{1}+\ldots+x_{k}\right)}$. Finally, we have
\begin{equation*}
    3 n^{-\delta_{0}\left(x_{1}+\ldots+x_{k}\right)} \leq n^{-\delta_{0}\left(x_{1}+\ldots+x_{k}\right) / 2},
\end{equation*}
 where $nps^2,nqs^2 = \omega(1)$. Thus, the proof is complete.
\end{proof}

\subsection{Proof of Theorem \ref{thm:matching achievability}}
   Note that $ns^{2}(p+(k-1)q) / k> (1+\varepsilon) \log n$ holds from the assumption \eqref{eq:thm matching degree}. We can choose a small enough $\epsilon$  such that it satisfies $1+\varepsilon > (1+\epsilon)(1-\epsilon)^{-2}$.
Let $\delta=\epsilon/2 \wedge \frac{\epsilon^{2}(1-\epsilon)nps^{2}}{100 k \log n} \wedge \frac{\epsilon^{2}(k-1)nqs^{2}}{100 k \log n}$ given by Lemma \ref{lem:prob bound 2}. Thus, $\delta \leq \epsilon$. Note that
\begin{equation}
\begin{aligned}
 \mathbb{P}\left(\hat{\pi} \neq \pi_{*}\right)&=\mathbb{P}\left(\hat{\tau} \neq \tau_{*}\right)=\mathbb{E}\left[\mathbb{P}\left(\hat{\tau} \neq \tau_{*} \mid \sigma, \tau_{*}\right)\right] \\
 &\leq \mathbb{E}\left[\mathbb{P}\left(\hat{\tau} \neq \tau_{*} \mid \sigma, \tau_{*}\right) 1\left(\mathcal{F}_{\epsilon}\right)\right]+\mathbb{P}\left(\mathcal{F}_{\epsilon}^{c}\right) .   
\end{aligned}
\end{equation}
By Lemma \ref{lem:balanced community}, we have that $\P(\mathcal{F}_{\epsilon}) \rightarrow 0$, where $k^{2}\log k =o(n)$. Thus we will choose $k<n^{1/3}$. Then, it remains to show $\mathbb{E}\left[\mathbb{P}\left(\hat{\tau} \neq \tau_{*} \mid \sigma, \tau_{*}\right) 1\left(\mathcal{F}_{\epsilon}\right)\right] \rightarrow 0$ as $n \to \infty$.

It cannot happen that only one node is mismatched. Thus, if $\hat{\tau}\in S_{x_1,\ldots,x_k}$ for some $x_1,\ldots,x_k$ and $\hat{\tau}\neq \tau_*$, then $x_1+\ldots+x_k \geq 2 $ must hold. By Lemma \ref{lem:prob bound 1} and \ref{lem:prob bound 2}, we can get
\begin{equation}
\begin{aligned}
    &\mathbb{P}\left(\hat{\tau} \neq \tau_{*} \mid \sigma, \tau_{*}\right) 1\left(\mathcal{F}_{\epsilon}\right)\\
    &=\sum_{x_{1}+\ldots+x_{k} \geq 2} \mathbb{P}\left(\hat{\tau} \in S_{x_{1},\ldots,x_{k}} \mid \sigma, \tau_{*}\right) 1\left(\mathcal{F}_{\epsilon}\right)\\
    &\leq \sum_{x_{1}+\ldots+x_{k} \geq 2} n^{-\delta\left(x_{1}+\ldots+x_{k}\right)},
\end{aligned}
\end{equation}
where $k \leq C(\sqrt{nps^{2}}\wedge nqs^{2})$ for a small enough constant $C$.
Note that there are  $k+\ell-1 \choose \ell$ different pairs of $\left(x_{1},\ldots,x_{k}\right)$ such that $x_{1}+\ldots+x_{k}=\ell$. Hence, we obtain
\begin{equation}
    \begin{aligned}
    \sum_{x_{1}+\ldots+x_{k} \geq 2} n^{-\delta\left(x_{1}+\ldots+x_{k}\right)} &\leq \sum_{\ell=2}^{\infty} {\ k+\ell-1 \choose \ell} n^{-\delta \ell}\\
    &=n^{-2 \delta} \sum_{\ell=0}^{\infty}{\ k+\ell+1 \choose \ell+2} n^{-\delta \ell}.
\end{aligned}
\end{equation}
We have 
\begin{equation}\label{eq:ratio}
    \frac{{\ k+\ell+2 \choose \ell+3} n^{-\delta (\ell+1)}}{{\ k+\ell+1 \choose \ell+2} n^{-\delta \ell}} = \frac{k+\ell+2}{\ell+3}n^{-\delta} \leq \frac{k+2}{3}n^{-\delta}\rightarrow 0
\end{equation}
if $k=o (n^{\delta})$. By using the sum of a geometric sequence, 
\begin{equation}
    \sum_{x_{1}+\ldots+x_{k} \geq 2} n^{-\delta\left(x_{1}+\ldots+x_{k}\right)} \leq n^{-2\delta} 2(k+1)^{2} \rightarrow 0
\end{equation}
if $k=o (n^{\delta})$. For any $M>0$, $k\leq n^{M nps^2/k \log n}$ holds since $k= O(\sqrt{nps^2})$ and $nps^2=\omega(1)$, and $k \leq n^{Mnqs^2/\log n}$ holds since $k=O(nqs^2)$ and $nqs^2= \omega(1)$. Thus, $k=o(n^\delta)$ holds, where $k \leq C(\sqrt{nps^2} \wedge nqs^2 \wedge n^c)$ for a sufficiently small $c$ and $C$.

Therefore, we can get
$$
\mathbb{E}\left[\mathbb{P}\left(\hat{\tau} \neq \tau_{*} \mid \sigma, \tau_{*}\right) 1\left(\mathcal{F}_{\epsilon}\right)\right] \stackrel{(a)}{\rightarrow} 0 \text{ as } n \rightarrow \infty
$$
where $(a)$ holds since $k\leq C(\sqrt{nps^{2}} \wedge nqs^{2} \wedge n^c)$ for sufficiently small constants $c$ and $C$.

\section{Proof of Theorem \ref{thm:matching impossible}}
Recall that $A$ and $B$ are the adjacency matrices of graphs $G_1$ and $G_2$, respectively. Recall the MAP estimator
\begin{equation}\label{eq:MAP estimator}
\hat{\pi}_{\mathrm{MAP}} \in \underset{\pi \in \mathcal{S}_{n}}{\arg \max } \mathbb{P}\left(\pi_{*}=\pi \mid A, B, \sigma\right).
\end{equation}
We will use the similar method as that of \cite{RS21} to prove the impossibility for exact matching, but we extend the result to more general $p, q,s$ and $k$ regimes.
Even if there are $k$ communities, some lemmas proposed in \cite{RS21} can still be used interchangeably. We cite only the main lemmas below. 

Given $A$ and $B$, for $\pi \in \mathcal{S}_{n}$, define the set
$$
T^{\pi} :=\left\{i \in[n]: \forall j \in[n], A_{i, j} B_{\pi(i), \pi(j)}=0\right\}
$$
as well as $T_{r}^{\pi}:=T^{\pi} \cap V_{r}$ for all $r\in[k]$. We can bound the MAP estimator as follows:

\begin{lemma}[Extension of Lemma 3.4 in \citep{RS21}] \label{lem:MAP estimator}
	For sufficiently large $n$ and for any permutation $\pi \in \mathcal{S}_{n}$, we obtain that
	\begin{equation}
		\mathbb{P}\left(\hat{\pi}_{\mathrm{MAP}}=\pi \mid A, B, \sigma\right) \leq \prod^{k}_{r=1}\frac{1}{\left|T_{r}^{\pi}\right| ! }.
	\end{equation}
\end{lemma}
Let us define
\begin{equation}
    \P_{\pi}(\cdot):=\P(\cdot|\pi_{*}=\pi),
\end{equation}
where $\pi \in S_{n}$. 
Furthermore, let us use $\E_\pi$ and $\operatorname{Var}_\pi$ to represent the expectation and variance operators that correspond to the probability measure $\P_\pi$.

\begin{lemma}[Extension of Lemma 3.5 in \citep{RS21}]\label{lem:T bound}
	For $\epsilon \in (0,1)$, suppose that $ns^{2}(p+(k-1)q) / k< (1-\epsilon) \log n$. Then there exists $\delta >\gamma>0$ such that if $k=o\left( n^{\delta -\gamma} \right)$, then
\begin{equation}\label{eq:isolated}
	\lim _{n \rightarrow \infty} \min _{\pi \in \mathcal{S}_{n}} \mathbb{P}_{\pi}\left(\left|T_{r}^{\pi}\right| \geq n^{\gamma}/2  \; \exists r\in[k]\right)=1.
\end{equation}
\end{lemma}

 \subsection{Proof of Theorem \ref{thm:matching impossible}}
      The following equation can be obtained. 
\begin{equation}\label{eq:MAP bound}
\begin{aligned}
    	\mathbb{P}\left(\hat{\pi}_{\mathrm{MAP}}=\pi_{*}\right)&=\frac{1}{n !} \sum_{\pi \in \mathcal{S}_{n}} \sum_{A, B, \sigma} 
 \mathbb{P}\left(A, B, \sigma \mid \pi_{*}=\pi\right) \\
     &\quad \times \mathbb{P}\left(\hat{\pi}_{\mathrm{MAP}}=\pi \mid A, B, \sigma, \pi_{*}=\pi\right) \\
     &\leq \frac{1}{n !} \sum_{\pi \in \mathcal{S}_{n}} \mathbb{E}_{\pi}\left[\prod^{k}_{r=1}\frac{1}{\left|T_{r}^{\pi}\right| !}\right].
\end{aligned}
\end{equation}
The last inequality holds by  Lemma \ref{lem:MAP estimator}, and $\mathbb{P}\left(\hat{\pi}_{\mathrm{MAP}}=\pi \mid A, B, \sigma, \pi_{*}=\pi\right)=$ $\mathbb{P}\left(\hat{\pi}_{\text {MAP }}=\pi \mid A, B, \sigma\right)$ holds since the MAP estimator depends only on $A,B$ and $\sigma$.

Let us define an event 
\begin{equation}
    \mathcal{R}_{\pi}:=\left\{\left|T_{r}^{\pi}\right| \geq n^{\gamma}/2  \; \exists r\in[k]\right\}.
\end{equation}
Then, we can get 
\begin{equation}\label{eq:Epi}
	\mathbb{E}_{\pi}\left[\prod^{k}_{r=1}\frac{1}{\left|T_{r}^{\pi}\right| ! }\right] \leq \frac{1}{\left(n^{\gamma}/2\right) !
 }+\mathbb{P}_{\pi}\left(\mathcal{R}_{\pi}^{c}\right).
\end{equation}
By combining \eqref{eq:Epi} and \eqref{eq:MAP bound}, we obtain
\begin{equation}
    \begin{aligned}
        \mathbb{P}\left(\hat{\pi}_{\mathrm{MAP}}=\pi_{*}\right) &\leq \frac{1}{(n ^{\gamma }/2)!}+\frac{1}{n !} \sum_{\pi \in \mathcal{S}_{n}} \mathbb{P}_{\pi}\left(\mathcal{R}_{\pi}^{c}\right)\\
        &\leq \frac{1}{(n ^{\gamma}/2) !}+\max _{\pi \in \mathcal{S}_{n}} \mathbb{P}_{\pi}\left(\mathcal{R}_{\pi}^{c}\right) .
    \end{aligned}
\end{equation}
Let $\delta=1-(1+\epsilon / 2)^{2} ns^{2} \frac{(p+(k-1)q)}{ k \log n}$ and choose $\gamma$ less than $\delta$. Let $t$ be less than $\delta-\gamma$. By Lemma \ref{lem:T bound}, finally we get 
\begin{equation}
    \mathbb{P}\left(\hat{\pi}_{\mathrm{MAP}}=\pi_{*}\right) \rightarrow 0.
\end{equation}

\subsection{Proof of Lemma \ref{lem:T bound}}\label{proof of lemma T bound}
\begin{proof}
    Recall that 
$\mathcal{E}^{+}(\sigma):=\left\{(i, j) \in \mathcal{E}: \sigma_{i} = \sigma_{j}\right\}$ and $\mathcal{E}^{-}(\sigma):=\left\{(i, j) \in \mathcal{E}: \sigma_{i} \neq \sigma_{j}\right\}$.
For a fixed  $\pi \in \mathcal{S}_{n} $, we assume that $\sigma$ is given and $\pi_*=\pi$. Then, we obtain that
$$
A_{i, j} B_{\pi(i), \pi(j)} \sim \begin{cases}\text { Bern }\left(s^{2}p\right) & \text { if }(i, j) \in \mathcal{E}^{+}(\sigma), \\ \text { Bern }\left(s^{2} q\right) & \text { if }(i, j) \in \mathcal{E}^{-}(\sigma) .\end{cases}
$$
In addition, for a fixed $i \in [n]$, $A_{i, j} B_{\pi(i), \pi(j)}$ are mutually independent across $j\in [n]\backslash\{ i\}$. Thus, for $i \in V_1$, we obtain

\begin{equation}\label{eq:P}
    \mathbb{P}_{\pi}\left\{i \in T_1^{\pi} \mid \sigma\right\}=\left(1-s^{2}p\right)^{\left|V_{1}\right|-1}\left(1-s^{2} q\right)^{|V_{2}|+\ldots|V_{k}|}.
\end{equation}

On the event $\mathcal{F}_{\epsilon}$, we have $\left|V_{r}\right| \leq(1+\epsilon / 2) n / k$ for all $r \in [k]$. Thus, we have
\begin{equation}\label{eq:P1}
\begin{aligned}
    \log (1-sp^{2})^{|V_{1}|-1}&\geq \left(1+\frac{\epsilon}{2}\right) \frac{n}{k}\log \left(1-s^{2} p\right)\\
    &\stackrel{(a)}{\geq}  -\left(1+\frac{\epsilon}{2}\right)^{2} \frac{n}{k}s^{2}p
\end{aligned}
\end{equation}
and
\begin{equation}\label{eq:P2}
\begin{aligned}
    \log (1-sq^{2})^{|V_{2}|+\ldots+|V_{k}|}&\geq \left(1+\frac{\epsilon}{2}\right) \frac{(k-1)n}{k}\log \left(1-s^{2} q\right)\\
    &\stackrel{(b)}{\geq}  -\left(1+\frac{\epsilon}{2}\right)^{2} \frac{(k-1)n}{k}s^{2}q.
\end{aligned}
\end{equation}
The inequality $(a)$ and $(b)$ hold since $\log (1-t)  \geq -(1+\epsilon/2)t$ for all small enough $t>0$.

 For $\epsilon \in(0,1)$, let us define
$$
\begin{aligned}
	\delta:=1-(1+\epsilon / 2)^{2} ns^{2} \frac{(p+(k-1)q)}{ k \log n} .
\end{aligned}
$$
By assumption $ ns^{2}(p+(k-1)q) / k< (1-\epsilon) \log n$, $\delta>0$ holds.
From \eqref{eq:P}, \eqref{eq:P1}, \eqref{eq:P2} and the definition of $\delta$, we can get 
$$
\begin{aligned}
	\log \mathbb{P}_{\pi}\left(i \in T_1^{\pi} \mid \sigma\right) \geq (\delta-1) \log n.
\end{aligned}
$$
It means that  $\mathbb{P}_{\pi}\left(i \in T_1^{\pi} \mid \sigma\right) \geq n^{\delta-1}$, where $\mathcal{F}_{\epsilon}$ holds. Therefore, we can obtain
\begin{equation}\label{eq:exp T1}
	\mathbb{E}_{\pi}\left[\left|T_{1}^{\pi}\right| \mid \sigma\right] \geq\left|V_{1}\right| n^{\delta-1} \geq \frac{1-\epsilon / 2}{k} n^{\delta} \geq \frac{1}{2k} n^{\delta}.
\end{equation}
For $i \in V_1$, let us define an indicator variable as follow:
\begin{equation}
    X_i = \begin{cases}
          1 & \text { if }  i \in T_{1}^{\pi},\\
         0& \text { if }  i \notin T_{1}^{\pi}.
    \end{cases}
\end{equation}
Then, we have
\begin{equation}\label{eq:Var T1}
\begin{aligned}
 \operatorname{Var}_{\pi}\left(\left|T_{1}^{\pi}\right| \mid \sigma\right)&=\operatorname{Var}_{\pi}\left(\sum_{i \in V_{1}} X_{i} \mid \sigma\right)\\
 &=\sum_{i \in V_{1}} \sum_{j \in V_{1}}\mathbb{E}_{\pi}\left[X_{i}X_{j} \mid \sigma\right]- \left(\mathbb{E}_{\pi} \left[|T^{\pi}_{1}| \mid \sigma \right] \right)^{2},
\end{aligned}	
\end{equation}
and 
\begin{equation}
    \begin{aligned}
        \mathbb{E}_{\pi}[X_{i}X_{j} \mid \sigma] &= (1-ps^{2})^{2|V_{1}|-3}(1-qs^{2})^{2\left(|V_{2}|+\ldots+|V_{k}|\right)}\\
        &=\frac{ \left(\mathbb{E}_{\pi}[X_{i}\mid \sigma] \right)^{2}}{1-ps^{2}}\\
        &=\frac{ \left(\mathbb{E}_{\pi}[\left|T_{1}^{\pi} \right| \mid \sigma] \right)^{2}}{|V_{1}|^{2}(1-ps^{2})}.
    \end{aligned}
\end{equation}
Therefore, we can get
\begin{equation}\label{eq:Var}
    \begin{aligned}
        \operatorname{Var}_{\pi}\left(\left|T_{1}^{\pi}\right|\right. & \left. \mid \sigma\right) =  \frac{ |V_{1}|(|V_{1}|-1)\left(\mathbb{E}_{\pi}[\left|T_{1}^{\pi} \right| \mid \sigma] \right)^{2}}{|V_{1}|^{2}(1-ps^{2})} \\
        &+\mathbb{E}_{\pi}[\left|T_{1}^{\pi} \right| \mid \sigma]  - \left(\mathbb{E}_{\pi} \left[|T^{\pi}_{1}| \mid \sigma \right] \right)^{2}\\
        &= \mathbb{E}_{\pi}[\left|T_{1}^{\pi} \right| \mid \sigma] +  \frac{|V_{1}|ps^{2}-1}{|V_{1}|(1-ps^{2})}\left(\mathbb{E}_{\pi} \left[|T^{\pi}_{1}| \mid \sigma \right] \right)^{2}.
    \end{aligned}
\end{equation}
Applying Chebyshev's inequality, we can have
\begin{equation}
\begin{aligned}
     \P_{\pi}\left\{|T^{\pi}_{1}| \leq \frac{1}{2}\mathbb{E}_{\pi}\left[\left|T_{1}^{\pi}\right| \mid \sigma\right]\right\} &\leq 4 \frac{ \operatorname{Var}_{\pi}\left(\left|T_{1}^{\pi}\right| \mid \sigma\right) }{\left(\mathbb{E}_{\pi}\left[\left|T_{1}^{\pi}\right| \mid \sigma\right]\right)^{2}}\\
     &\stackrel{(a)}{=}\frac{1}{\mathbb{E}_{\pi}(\left|T_{1}^{\pi} \right| \mid \sigma)} +\frac{|V_{1}|ps^{2}-1}{|V_{1}|(1-ps^{2})} \\
     &\leq \frac{1}{\mathbb{E}_{\pi}(\left|T_{1}^{\pi} \right| \mid \sigma)} +\frac{ps^{2}}{1-ps^{2}} \\
     &\leq \frac{2k}{n^{\delta}}+\frac{ps^{2}}{1-ps^{2}}.
\end{aligned}   
\end{equation}
The equality $(a)$ holds from  \eqref{eq:Var}.
Therefore, for any $\gamma < \delta$, if $k=o(n^{\delta-\gamma})$ and $\frac{ps^{2}}{1-ps^{2}}=o(1)$ then we have 
\begin{equation*}
\begin{aligned}
      \P_{\pi} \left\{|T^{\pi}_{1}| \geq \frac{1}{2} n^{\gamma}\right\} &\geq \E \left.\left[ \P_{\pi} \left\{|T^{\pi}_{1}| \geq \frac{1}{2} n^{\gamma} \right| \sigma \right\}  \mathrm{1}(\mathcal{F}_{\epsilon})\right]\\
      & \geq (1-2n^{-\gamma} -2ps^2) \P\{\mathcal{F}_{\epsilon}\} \rightarrow 1.
\end{aligned}
  \end{equation*}
Note that $\P\{\mathcal{F}_{\epsilon}\} \rightarrow 1$ holds by Lemma \ref{lem:balanced community}. For all $\pi \in S_n$, this lower bound holds. Thus, we get
\begin{equation}
     \min _{\pi \in \mathcal{S}_{n}} \mathbb{P}_{\pi}\left(\left|T_{r}^{\pi}\right| \geq n^{\gamma} \; \exists r\in[k]\right) \rightarrow 1
\end{equation}
as $n\rightarrow \infty$. 

\end{proof}

\section{Proof of Corollary \ref{cor:community recovery}}
Let us define the union graph $H:=G_1\cup G'_2$. We can see that  $H\sim$SBM$\left(n,p(1-(1-s)^2),q(1-(1-s)^2),k \right) $. We generate the graph $G_1 \vee_{\pi} G_2$ in the following way. 
\begin{itemize}
    \item If $(i,j)$ is an edge in $G_1$ or $(\pi(i),\pi(j))$ is an edge in $G_2$, $(i,j)$ becomes the edge in $G_1 \vee_{\pi} G_2$. 
\end{itemize}
For notational simplicity, we will use $H_{\pi}$ instead of $G_1 \vee_{\pi} G_2$. Then, we can see that $H=H_{\pi_*}$. Let $\hat{\pi}$ be the estimator that achieves the information theoretic limit for graph matching as the result of Theorem \ref{thm:matching achievability} and $\hat{\sigma}$ be the estimator that achieves information theoretic limit for exact community recovery as the result of Theorem 1 in \citep{ABKK17}. Then, we have
\begin{equation}
    \begin{aligned}
        &\P\left(\hat{\sigma}(H_{\hat{\pi}})\right) \neq \sigma) \\
        &\leq \P\left(\{\hat{\sigma}(H_{\hat{\pi}}) \neq \sigma  \} \text{ and }\{H_{\hat{\pi}}=H\}  \right)+\P\left(H_{\hat{\pi}} \neq H  \right)\\
        & \leq \P\left(\hat{\sigma}(H)\neq \sigma  \right)+\P\left(\hat{\pi} \neq \pi_{*}  \right).
    \end{aligned}
\end{equation}
We have $\P\left(\hat{\sigma}(H) \neq \sigma  \right) \to 0$ when  \eqref{eq:cor:community recovery condition} holds since $H\sim$SBM$\left(n,p(1-(1-s)^2),q(1-(1-s)^2),k \right) $ and the result of \citep{ABKK17}. We can also obtain $\P\left(\hat{\pi} \neq \pi_{*} \right) \to  0$ when \eqref{eq:cor:community recovery degree} holds since Theorem \ref{thm:matching achievability}. Thus, the proof is complete.

\section{Proof of Corollary \ref{cor:community recovery impossible} }
Let $H\sim$SBM$\left(n,p(1-(1-s)^2),q(1-(1-s)^2),k \right) $ and let us define $\sigma_{H}$ as the community labels of the graph $H$. For any estimator $\tilde{\sigma}$, we can obtain
\begin{equation}\label{eq:cor4 imposs1}
    \lim _{n \rightarrow \infty} \P\left(\tilde{\sigma}(H)=\sigma_{H} \right)=0
\end{equation}
when \eqref{eq:cor:community recovery impossible condition} holds since Theorem 2 in \citep{ABKK17}.

We construct two graphs $H_1$ and $H'_2$ in the following way.
\begin{itemize}
    \item If $(i,j)$ is not an edge in the graph $H$, it is also not an edge in both $H_1$ and $H'_2$
    \item If $(i,j)$ is an edge in the graph $H$, then
\end{itemize}
\begin{equation*}
    \begin{cases}
    & (i,j) \text{ is an edge in } H_1 \text{ and } H'_2 \text{ with probability } r_{11} ;\\
    &(i,j) \text{ is an edge in } H_1 \text{ but not an edge in } H'_2 \\ & \text{ with probability } r_{10};\\
    & (i,j) \text{ is an edge in } H'_2 \text{ but not an edge in } H_1 \\ & \text{ with probability } r_{01},
    \end{cases}
\end{equation*}
where 
\begin{equation}
    \begin{aligned}
        r_{11} = \frac{s^2 }{1-(1-s)^2},\quad r_{10}=r_{01} = \frac{s(1-s)}{1-(1-s)^2}.
    \end{aligned}
\end{equation}
It can be seen that $(H_1,H'_2,\sigma_{H})$ and $(G_1,G'_2,\sigma)$ have the same distribution. Let $H_2$ be the graph obtained by permuting the vertices of $H'_2$ by permutation $\pi \in S_n$. Then we can see that  $(H_1,H_2,\sigma_{H})$ and $(G_1,G_2,\sigma)$ also have the same distribution. Suppose that there exists an estimator $\hat{\sigma}$ such that 
\begin{equation}\label{eq:cor4 result}
     \limsup_{n \rightarrow \infty} \P\left(\hat{\sigma}(G_1,G_2)=\sigma \right)>0.
\end{equation}
Since $(H_1,H_2,\sigma_{H})$ and $(G_1,G_2,\sigma)$ have the same distribution, it follows that 
\begin{equation}\label{eq:cor4 imposs2}
     \limsup_{n \rightarrow \infty} \P\left(\hat{\sigma}(H_1,H_2)=\sigma_{H} \right)>0.
\end{equation}
Two graphs $H_1$ and $H_2$ are constructed from the graph $H$ through random sampling and random permutation, which is why \eqref{eq:cor4 imposs1} and \eqref{eq:cor4 imposs2} are contradictory. Thus, the proof is complete.

    \section{Proof of Corollary \ref{cor:multiple achieve} and \ref{cor:multiple impossible}}
Corollary \ref{cor:multiple achieve} can be proved by combining the proof of Theorem 1.6 in \citep{RS21} and the proof of Corollary \ref{cor:community recovery}. Corollary \ref{cor:community recovery impossible} can also be proved by combining the proof of Theorem 1.7 in \citep{RS21} and the proof of Corollary \ref{cor:community recovery impossible}.

\end{document}